\documentclass[sigconf,nonacm]{acmart}

\setcopyright{none}
\acmDOI{}
\usepackage{booktabs}
\usepackage{pgfplots}
\pgfplotsset{compat=1.15}

\newtheorem{proposition}{Proposition}
\newtheorem{thm}{Theorem}
\newtheorem{corollary}{Corollary}

\newcommand{\rstar}{0.35}
\newcommand{\taum}{3.0}

\begin{document}

\title{Ride, Raise, or Wake}
\subtitle{A Theory of Light-Thread Placement on Multi-Tier Processors}

\author{Mohamed Amine Bergach}
\email{mbergach@gmail.com}
\affiliation{%
  \institution{Illumina}
  \city{San Diego}
  \state{California}
  \country{USA}}
\renewcommand{\shortauthors}{Bergach}

\begin{abstract}
Phones and laptops now mix several kinds of cores, and each kind sits in its own voltage and
clock domain. When a light thread arrives, should it run beside heavy threads on a fast cluster
that is already awake, or wake a cluster of efficient cores? We answer with three costs. A thread
can \emph{ride} an awake domain at the voltage other threads hold, \emph{raise} that voltage for
everyone, or \emph{wake} a sleeping domain and pay its fixed power. From these costs follows a
simple rule: ride while the thread's demand is below $r^*=\Phi/(e_h-e_l)$, the fixed power of the
cluster it would wake over the gap in energy per operation; otherwise wake the smallest tier that
can serve it. The best tier is not monotone in demand: tiny threads ride high, medium threads wake
a low tier, heavy threads go high again. Two common designs err on opposite sides of the rule:
grouping a light thread with the heavy threads of its application overpays above the threshold,
and an energy model without idle costs overpays below it. In a model with five tiers, grouping
costs 80 to 145\% more energy than the rule, and more as tiers are added. Measurements on an
Apple M6 anchor the parameters.

\end{abstract}

\maketitle

\section{Introduction}
\label{sec:intro}

Phones and laptops now mix several kinds of cores. Apple's M6 has efficiency, performance and
super cores; Intel's client parts add low-power efficiency cores; phone chips add a prime core. We
call each kind a \emph{tier}. Each tier sits in a voltage and clock domain, a \emph{cluster}, and
sometimes two tiers share one: on the M6 the performance and super cores share a cluster. A
scheduler knows where a heavy thread belongs (a fast core) and where background work belongs (an
efficient one). The hard case is the light thread, whose demand is small but not zero.

Where it runs matters. On the M6 we placed a light thread next to two busy threads of its own
application. macOS ran it on the fast cluster, at the voltage the busy threads held, and it used
2.1 to 2.2 times the energy per operation it used when we launched it as an application of its own,
which macOS ran on the efficiency cores. When the fast cluster had room to rise, the light thread
raised its voltage for every thread there and paid up to 3.9 times. Yet riding a busy cluster is
not always wrong: waking an idle cluster has a fixed price, and a thread light enough should not
pay it.

We make this trade exact with a small model and three costs. A thread can \emph{ride} an awake
cluster at the operating point others hold, \emph{raise} that point, or \emph{wake} a sleeping
cluster. Our contributions:
\begin{itemize}
\item A threshold rule: ride while the demand is below $r^*=\Phi/(e_h-e_l)$, with a dimensionless
form and a version for bursty threads (Section~\ref{sec:costs}).
\item A structure theorem: the best tier is not monotone in demand, exactly when a larger tier is
held awake by others while a cheaper one sleeps (Section~\ref{sec:tiers}).
\item Closed forms for many tiers and for the losses of two common policies, which err on opposite
sides of the rule (Sections~\ref{sec:tiers} and~\ref{sec:policies}).
\item A five-tier model and anchors from measurements on an M6 (Section~\ref{sec:anchors}).
\end{itemize}

Parts of this are known. Android's energy-aware scheduler and Qualcomm's have charged for waking
an idle cluster~\cite{rasmussen2014eas,qcom_hmp_powercost_src}, Linux has packed small
tasks~\cite{guittot2013packing}, Linux charges a thread for raising a domain~\cite{linuxeas}, and
Apple names the cost of riding a busy cluster a ``passenger tax''~\cite{xnu_edgesched2025}. What
we add is the closed-form rule, the shape of the decision, and what the common policies lose.

\section{A small model}
\label{sec:model}

\paragraph{Clusters and tiers.} A cluster is a group of cores that share one voltage $V$ and one
clock $f$. Its voltage follows a table $V(f)$ with a floor $V_0$ at its lowest clock $f_0$, and its
clock has a ceiling $f_{\max}$. A cluster may hold cores of more than one design. It is either
asleep, and costs nothing, or awake, and pays a fixed \emph{overhead} $\Phi(V)$ for its caches,
clock tree and links, whatever its load.

\paragraph{Energy per operation.} A busy core spends, per operation,
\begin{equation}
e(V,f) = c\,V^2 + \frac{s(V)}{I f},
\label{eq:e}
\end{equation}
where $c$ is the capacitance switched per operation, $I$ the operations per cycle and $s(V)$ the
static power of a busy core. An idle core turns itself off.

\paragraph{Governor.} A cluster runs at the clock its most demanding thread needs,
$f=\max\bigl(f_0,\ \max_j r_j/(u I_j)\bigr)$, where $r_j$ is thread $j$'s demand in operations per
second and $u\le1$ a utilization target: $u=1$ runs just fast enough, a smaller $u$ keeps headroom.
macOS behaves like $u\approx0.6$ to $0.8$ in our runs, and Linux's frequency governor uses
$0.8$~\cite{schedutil_src2026}.

\paragraph{The light thread.} The thread we place has demand $r$ and gets a core of its own. The
threads already in a cluster keep their work; at the cluster's operating point they draw
$D(V,f)=\sum_j r_j e_j(V,f)$. Table~\ref{tab:symbols} lists the symbols.

\begin{table}[t]
\caption{Symbols.}
\label{tab:symbols}
\centering\small
\begin{tabular}{ll}
\toprule
$r$ & demand of the light thread (operations per second)\\
$e(V,f)$ & energy per operation of a busy core, (\ref{eq:e})\\
$\Phi$ & overhead power of an awake cluster\\
$u$ & utilization target of the governor\\
$e_h$, $e_l$ & energy per operation riding the busy cluster, and\\
 & on the efficient cluster at its floor\\
$r^*$ & threshold demand\\
$\omega$ & $\Phi_l/(r_l e_l)$: overhead over one core's power\\
$\tau$ & $e_h/e_l$: the passenger tax\\
\bottomrule
\end{tabular}
\end{table}

\section{Ride, raise, or wake}
\label{sec:costs}

\begin{proposition}[Three costs]
\label{prop:costs}
Adding the light thread to a core of design $d$ in cluster $C$ raises the power by
\begin{itemize}
\item \emph{ride}, if $C$ is awake and its clock already serves the thread ($r\le uI_d f_C$):
  $\Delta P=r\,e_d(V_C,f_C)$;
\item \emph{raise}, if $C$ is awake but must speed up, to $f'=r/(uI_d)$ and $V'=V_C(f')$:
  $\Delta P=r\,e_d(V',f')+\bigl[D_C(V',f')-D_C(V_C,f_C)\bigr]+\bigl[\Phi_C(V')-\Phi_C(V_C)\bigr]$;
\item \emph{wake}, if $C$ is asleep, at $f'=\max(f_0,r/(uI_d))$ and $V'=V_C(f')$:
  $\Delta P=r\,e_d(V',f')+\Phi_C(V')$.
\end{itemize}
\end{proposition}
\begin{proof}
A cluster's power is its threads' power at the shared operating point plus its overhead. The three
cases are the change before and after, depending on whether the operating point moves and whether
the overhead was already paid.
\end{proof}

The best place for the thread is where $\Delta P$ is least. The raise case is the one that is easy
to forget: the term $D(V')-D(V)$ is paid by the other threads. To first order it is about
$2D\,\Delta V/V$, so a light thread that lifts a crowded cluster by a few percent of its voltage
pays for everyone there.

\begin{proposition}[Ride or wake]
\label{prop:threshold}
Let cluster $h$ be awake with room for the thread, at energy per operation $e_h$, and let cluster
$l$ be asleep, with the thread fitting at $l$'s floor: energy $e_l$ and overhead $\Phi_l$. Riding
$h$ costs less than waking $l$ if and only if
\begin{equation}
r < r^* = \frac{\Phi_l}{e_h-e_l},
\label{eq:rstar}
\end{equation}
and always if $e_h\le e_l$.
\end{proposition}
\begin{proof}
The two costs are $r\,e_h$ and $r\,e_l+\Phi_l$.
\end{proof}

In units of cores the rule needs only two ratios. Let $r_l$ be what one core of $l$ serves at its
floor, $\omega=\Phi_l/(r_l e_l)$ the overhead over that core's power, and $\tau=e_h/e_l$ the
passenger tax. Then
\begin{equation}
\frac{r^*}{r_l}=\frac{\omega}{\tau-1}.
\end{equation}
Riding pays only when the tax is small or the cluster to wake is costly next to one of its cores
(Table~\ref{tab:threshold}). With the model's M6 values (Section~\ref{sec:anchors}: $e_h=213$\,pJ on a
performance core at the fast cluster's top voltage, where we measured 206, and $e_l=72$\,pJ on an
efficiency core at its floor), an overhead of 0.05\,W gives $r^*=\rstar$\,Gops/s, a few percent of an efficiency core.

\begin{table}[t]
\caption{The threshold $r^*$ in cores of the efficient cluster, $\omega/(\tau-1)$.}
\label{tab:threshold}
\centering\small
\begin{tabular}{rrrrr}
\toprule
& \multicolumn{4}{c}{tax $\tau$}\\
\cmidrule(lr){2-5}
$\omega$ & 1.5 & 2 & 3 & 5\\
\midrule
0.1 & 0.20 & 0.10 & 0.05 & 0.03\\
0.3 & 0.60 & 0.30 & 0.15 & 0.07\\
1 & 2.00 & 1.00 & 0.50 & 0.25\\
\bottomrule
\end{tabular}

\end{table}

\paragraph{Bursty threads.} Proposition~\ref{prop:threshold} assumes that the woken cluster stays
awake while the thread lives. A thread that runs in bursts lets it sleep in between.

\begin{proposition}[Bursty threads]
\label{prop:bursty}
Let the thread run $\nu$ bursts per second, and let cluster $l$ spend $E_w$ to wake up and fall
asleep, and stay awake for a time $t_h$ after each burst. If the gaps are longer than $t_h$, riding
$h$ costs less if and only if
\begin{equation}
\nu>\nu^*=\frac{r\,\bigl(e_h-e_l-\Phi_l/(I_l f_l)\bigr)}{E_w+\Phi_l t_h}.
\end{equation}
\end{proposition}
\begin{proof}
Waking $l$ costs $r\,e_l+\Phi_l\bigl(r/(I_l f_l)+\nu t_h\bigr)+\nu E_w$: each burst pays the
transition and the overhead for its busy time and its hold time.
\end{proof}

Frequent small bursts should ride, rare large ones should wake. This is the break-even reasoning of
dynamic power management~\cite{benini_dpm_tvlsi2000}, applied to placement. When $\nu t_h\ge1$ the
cluster never sleeps and Proposition~\ref{prop:threshold} returns.

\begin{proposition}[Ride or ride]
\label{prop:ride}
Between two awake clusters with room for the thread, the cheaper one is the one with the lower
energy per operation at its current operating point, whatever the demand.
\end{proposition}
\begin{proof}
Both costs are $r$ times a constant.
\end{proof}

A large core held at a low voltage can therefore beat a small core held at a high one. Linux's
energy model already makes this comparison~\cite{linuxeas}.

\section{Which tier?}
\label{sec:tiers}

With one thread and nothing else running, the best tier rises with demand: small cores for light
work, large ones for heavy work. This staircase follows from a standard argument whenever larger
tiers gain on smaller ones as demand grows~\cite{milgrom_mcs_econometrica1994}. Busy clusters
break it.

\begin{thm}
\label{thm:shape}
(a) If every cluster is asleep, or every cluster is awake at its own floor, the best tier is
nondecreasing in demand (under the staircase condition above).
(b) If a larger cluster $h$ is held awake by other threads while a smaller cluster $l$ with
$e_l<e_h$ sleeps, and $l$ can serve more than $r^*$, the best tier is not monotone in demand: $h$
below $r^*$, a smaller tier above $r^*$, and a tier at least as large as $h$ again for demand that
no smaller tier can serve.
\end{thm}
\begin{proof}
(a) In both states each tier's cost is its single-thread cost plus a constant (its own overhead
when all clusters sleep, nothing when all are awake), and constants do not change how costs differ
as $r$ grows, so the staircase stays. (b) The held cluster costs $r\,e_h$, with no constant; the
sleeping one costs $\Phi_l+r\,e_l$. Their difference, $r(e_h-e_l)-\Phi_l$, grows with $r$, so the
order flips at $r^*$ (Proposition~\ref{prop:threshold}). Above what $l$ can serve, only larger tiers
remain.
\end{proof}

Fixed wake costs alone never break the staircase. What breaks it is a larger tier that other
threads keep awake: its overhead is already paid, but its voltage is high.

\begin{corollary}[The top tier]
\label{cor:top}
The top tier is the right place for a light thread only if no cheaper cluster is awake with room
and the thread is lighter than the wake threshold of the cheaper sleeping clusters, or if the
thread is too heavy for every lower tier. A middle tier that is awake is the cheapest ride for
every light thread below its own threshold.
\end{corollary}

\paragraph{Many tiers.} Suppose energy per operation grows by a factor $b$ from one tier to the
next at the same voltage ($b$ is the square root of the area ratio under Pollack's rule). Riding
tier $h$ at voltage $V_h$ instead of tier $l$ at its floor $V_0$ then costs
$\tau=b^{\,h-l}(V_h/V_0)^2$ per operation. For an adjacent tier awake at the same voltage
$\tau=b$, and the threshold is $\omega/(b-1)$: with closely spaced tiers, light threads should ride
the nearest awake tier rather than wake a smaller one.

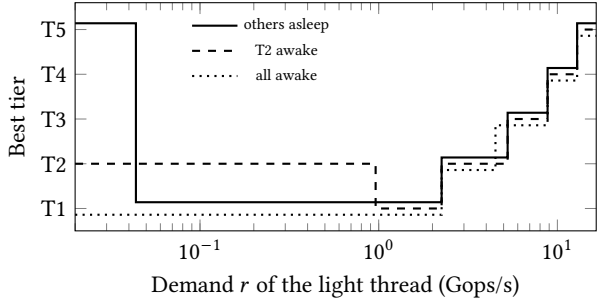
\begin{figure}[t]
\centering
\begin{tikzpicture}
\begin{axis}[width=\columnwidth, height=4.6cm, xmode=log, xmin=0.02, xmax=16.5,
  ymin=0.5, ymax=5.6, ytick={1,...,5}, yticklabels={T1,T2,T3,T4,T5},
  xlabel={Demand $r$ of the light thread (Gops/s)}, ylabel={Best tier},
  legend style={at={(0.36,0.98)}, anchor=north, font=\scriptsize, draw=none, fill=none}]
\addplot[const plot, thick, black] table[x=r, y expr=\thisrow{asleep}+0.14] {data/decision.dat};
\addlegendentry{others asleep}
\addplot[const plot, thick, dashed, black] table[x=r, y=ttwo] {data/decision.dat};
\addlegendentry{T2 awake}
\addplot[const plot, thick, dotted, black] table[x=r, y expr=\thisrow{awake}-0.14] {data/decision.dat};
\addlegendentry{all awake}
\end{axis}
\end{tikzpicture}
\caption{The best tier for a light thread on a five-tier processor while two heavy threads hold
the top tier, T5, at its top clock. With the other tiers asleep, tiny threads ride T5; with T2
awake they ride T2; with every tier awake the choice rises with demand. Lines are offset slightly
for legibility.}
\label{fig:decision}
\end{figure}

Figure~\ref{fig:decision} shows an example built on the M6's designs (Section~\ref{sec:anchors}):
T1 is smaller than the efficiency core, T2 is the efficiency core, T3 the performance core, T4
the super core and T5 larger, each in its own cluster, with overheads from 0.02\,W (T1) to
0.2\,W (T5). With the other tiers asleep, threads under 0.04\,Gops/s ride T5, then each tier takes
its turn as demand grows, and T5 returns for threads too heavy for T4. With T2 awake, light
threads ride T2 up to 0.94\,Gops/s. With every tier awake, the staircase returns.

\section{Two policies, two errors}
\label{sec:policies}

Two simple policies are in use. \emph{Grouping} puts a light thread with the heavy threads of its
application: macOS recommends a cluster per thread group and quality-of-service class, and an
application's threads share a group by default~\cite{xnu_edgesched2025,xnu_src2025}.
\emph{Ignoring idle costs} puts it wherever energy per operation is lowest, as an energy model does
when an idle cluster costs nothing. Mainline Linux removed idle costs from its energy model in
2018~\cite{eggemann2018eas}; in our reading of its code~\cite{feec_src2026} an empty cluster
is free, so it never rides.

\begin{proposition}[What each policy loses]
\label{prop:policies}
In the setting of Proposition~\ref{prop:threshold}, always riding loses nothing below $r^*$ and
$r(e_h-e_l)-\Phi_l$ above it; its cost over the optimum, $r e_h/(r e_l+\Phi_l)$, rises toward the
tax $\tau$. Always waking loses $\Phi_l-r(e_h-e_l)$ below $r^*$ and nothing above it; its ratio
$(r e_l+\Phi_l)/(r e_h)$ grows without bound as $r\to0$, but its loss never exceeds $\Phi_l$. For
demand spread evenly in log scale over $[r_a,r_b]$, with $L=\ln(r_b/r_a)$, the mean ratios are
\begin{align*}
\text{ride:}\quad & \Bigl[\ln\frac{r^*}{r_a}+\tau\ln\frac{r_b e_l+\Phi_l}{r^* e_l+\Phi_l}\Bigr]\Big/L,\\
\text{wake:}\quad & \Bigl[\frac{1}{\tau}\ln\frac{r^*}{r_a}+\frac{\Phi_l}{e_h}\Bigl(\frac{1}{r_a}-\frac{1}{r^*}\Bigr)+\ln\frac{r_b}{r^*}\Bigr]\Big/L.
\end{align*}
\end{proposition}
\begin{proof}
The losses follow from Proposition~\ref{prop:threshold}; the means integrate them against
$dr/(rL)$.
\end{proof}

\begin{figure}[t]
\centering
\begin{tikzpicture}
\begin{axis}[width=\columnwidth, height=4.4cm, xmode=log, ymode=log, xmin=0.02, xmax=5,
  ymin=0.9, ymax=14, ytick={1,2,4,8}, yticklabels={1,2,4,8},
  xlabel={Demand $r$ of the light thread (Gops/s)}, ylabel={Energy over the optimum},
  legend style={at={(0.98,0.98)}, anchor=north east, font=\scriptsize, draw=none, fill=none}]
\addplot[thick, black] table[x=r, y=ride] {data/losses.dat};
\addlegendentry{always ride (grouping)}
\addplot[thick, dashed, black] table[x=r, y=wake] {data/losses.dat};
\addlegendentry{always wake (no idle cost)}
\draw[dotted] (axis cs:\rstar,0.9) -- (axis cs:\rstar,14);
\node[font=\scriptsize, anchor=south west] at (axis cs:\rstar,1.05) {$r^*$};
\end{axis}
\end{tikzpicture}
\caption{What the two policies pay over the optimum on the M6 numbers ($e_h=213$\,pJ,
$e_l=72$\,pJ, $\Phi_l=0.05$\,W, so $r^*=\rstar$\,Gops/s and $\tau=\taum$).}
\label{fig:losses}
\end{figure}
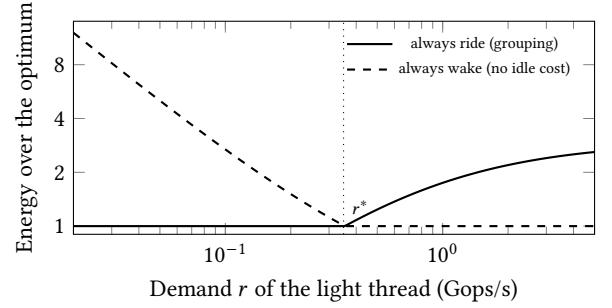

The two policies err on opposite sides of $r^*$ (Figure~\ref{fig:losses}). Grouping overpays for
every thread heavier than $r^*$, by up to the tax; ignoring idle costs overpays for every thread
lighter than $r^*$, by up to the overhead it ignores. For demand between 0.05 and 5\,Gops/s, they
use 51\% and 59\% more energy than the rule on average.

More tiers make grouping worse (Table~\ref{tab:kcost}). For ladders of two to five tiers spanning
the same designs, with two heavy threads holding the top tier at its top clock, grouping costs 80
to 145\% more than the rule when the other tiers sleep, and over 300\% when a middle tier is awake.
Each new tier is one more cheap place that the grouped thread does not use.

\begin{table}[t]
\caption{Grouping against the rule, for demand from 0.02 to 8\,Gops/s, on ladders of $K$ tiers
whose top tier two heavy threads hold at its top clock; ``top tier best'' is the share of demands
for which riding the top tier is right.}
\label{tab:kcost}
\centering\small
\begin{tabular}{rrrr}
\toprule
Tiers & \multicolumn{2}{c}{Others asleep} & Middle awake\\
\cmidrule(lr){2-3}\cmidrule(lr){4-4}
$K$ & grouping costs & top tier best & grouping costs\\
\midrule
2 & +80\% & 34\% & n/a\\
3 & +129\% & 17\% & +330\%\\
4 & +144\% & 13\% & +328\%\\
5 & +145\% & 13\% & +382\%\\
\bottomrule
\end{tabular}

\end{table}

\section{Anchors from an Apple M6}
\label{sec:anchors}

We set the model's parameters from measurements on a Mac mini with an Apple M6 under macOS 27,
taken without root from the per-cluster power estimates the chip's power manager publishes. The
six efficiency cores form one cluster (0.97 to 2.94\,GHz, 630 to 810\,mV); the two super and four
performance cores share another (1.44 to 4.79\,GHz, 725 to 1105\,mV). We switched a light FMA
thread on and off every 5\,s and took its energy from the change in cluster power over its work in
the same windows. From these runs: $c=170$\,pJ/V$^2$ per FMA on efficiency and performance cores
and 270 on super cores; 2.5 and 3.8 FMAs per cycle on efficiency and super cores (we assume 3.2 on
performance cores); $u=0.7$. The Mac was never idle, so overheads could not be measured; we vary
them.

Three observations check the model.
\begin{itemize}
\item \emph{Riding at a held voltage costs about the voltage ratio squared.} Beside two busy
threads that held the fast cluster at its top state, the light thread ran on performance cores
for 206\,pJ per FMA, against 94\,pJ on the efficiency cores when it was its own application: 2.19
times, where $(1105/727)^2$ predicts 2.31.
\item \emph{Raising costs more.} With the busy threads below the top state, the light thread
lifted the fast cluster's voltage by 16 to 66\,mV while it ran and paid up to 3.9 times its cost
on the efficiency cores, the raise case of Proposition~\ref{prop:costs}.
\item \emph{macOS moves a thread to the fast cluster when the efficient one runs out of headroom.}
A light thread in its own application stayed on the efficiency cores up to 4.5\,Gops/s and moved
by 6\,Gops/s; at $u=0.7$ the model puts the efficiency cluster's limit at 5.1\,Gops/s. The
efficiency cluster was awake in all our runs, so by Proposition~\ref{prop:threshold} it was the
right place for every light thread it could serve.
\end{itemize}

The same trade shows on a chip with a single voltage rail. On an AMD Ryzen 9 3900X, a memory-bound
thread keeps the idle cores of its core complex awake, at 0.4 to 0.8\,W each, which is an overhead
in our terms; filling one complex before opening another used 6.7\% less rail energy, as
Proposition~\ref{prop:threshold} predicts when the overhead is large and the gap in energy per
operation small.

\section{Related work}
\label{sec:related}

\paragraph{Energy-aware schedulers.} Linux's energy-aware scheduler estimates, for each candidate
cluster, the energy of placing a task there, including the cost of raising the cluster's operating
point for the tasks already running~\cite{linuxeas,feec_src2026}. Its first energy model,
kept in Android, also charged the busy and idle power of whole
clusters~\cite{rasmussen2014eas,aosp_eas49_src}; mainline Linux dropped idle costs in
2018~\cite{eggemann2018eas}. Qualcomm's scheduler added a static cluster power cost to bias tasks
away from idle clusters~\cite{qcom_hmp_powercost_src}, and an earlier Linux proposal packed small
tasks onto busy cores below a fixed load~\cite{guittot2013packing}. These systems make the
comparison of Proposition~\ref{prop:threshold} numerically or by heuristic; we give its closed
form, its bursty version, and the shape it implies.

\paragraph{Apple.} Apple's scheduler notes name the passenger tax and recommend a cluster per
thread group and quality-of-service class~\cite{xnu_edgesched2025}, and its performance controller
runs a shared cluster at the highest of its thread groups' requests~\cite{apple_clpc_patent2019}.
We measure the tax and place the grouping policy against the rule.

\paragraph{Theory.} Running a voltage island at the frequency of its most demanding core costs at
most a bounded factor over per-core frequencies~\cite{pagani_sfa_tecs2014}. Below a critical speed,
slowing down wastes energy~\cite{jejurikar2004}, and dynamic power management decides when to sleep
from a break-even time~\cite{benini_dpm_tvlsi2000}. Heterogeneous cores save energy by matching work to
core type~\cite{kumar_hetero_micro2003}, and studies of how many types to build find that two or
three capture most of the benefit~\cite{vancraeynest_designchoices_taco2013,azizi_marginal_isca2010}; ``non-monotonic'' core selection there
means core types that cannot be ranked~\cite{navada_nonmonotonic_pact2013}, not the non-monotone
placement of Theorem~\ref{thm:shape}.

\section{Conclusion}
\label{sec:conclusion}

Where a light thread should run comes down to three costs and one threshold. It should ride a
cluster that is already awake while its demand is below $r^*=\Phi_l/(e_h-e_l)$, and otherwise wake
the smallest tier that serves it. Because a busy cluster's overhead is already paid while its
voltage is high, the best tier is not monotone in demand, and the two common policies fail on
opposite sides of the threshold. More tiers make the choice matter more: in our five-tier model,
grouping costs 80 to 145\% more energy than the rule. A scheduler needs three numbers per
cluster to apply it: whether it is awake, its energy per operation where it is held, and its
overhead.

\paragraph{Limitations.} The model gives each thread its own core and a steady demand, except in
Proposition~\ref{prop:bursty}. The overheads are not measured on the M6, whose clusters never
idled, so the thresholds we quote are for a range of overheads. Our statement about mainline Linux
comes from reading its code, not from a measurement. On the M6 the scheduler sometimes ran the light
thread on the same cores as the busy ones, which our model does not capture.

\begin{acks}
We used a generative AI assistant (Claude, by Anthropic) to write measurement and analysis scripts,
to search the literature, and to draft and edit the text. We reviewed the measurements, the
analysis and the text, and we take responsibility for them.
\end{acks}

\bibliographystyle{ACM-Reference-Format}
\bibliography{refs}

\end{document}